\documentclass[12pt,journal,draftcls,onecolumn,twoside]{IEEEtran}
\usepackage{setspace}
\usepackage{amsmath}
\usepackage{amssymb}
\usepackage{graphicx}

\usepackage{cite}      

\usepackage{subfigure} 

\begin{document}

%
\title{Recovering Multiplexing Loss Through Successive Relaying Using Repetition Coding}
%
%
\author{Yijia~Fan, Chao~Wang, John~Thompson, H. Vincent~Poor
\thanks{Manuscript submitted June 2006, revised Feb 2007}
\thanks{Y. Fan's work was supported by the EPSRC Grant GR/S58782/01, UK and C. Wang's work was supported by the UK Mobile Virtual Centre of Excellence (VCE) Core 4 Project (www.mobilevce.com).}
\thanks{Y. Fan was with the Institute for Digital Communications,
University of Edinburgh. He is now with the Department of Electrical
Engineering, Princeton University, Princeton, NJ, 08544, USA.
(e-mail: yijiafan@princeton.edu)}
\thanks{C. Wang and J. Thompson are with the Institute for Digital Communications, University of Edinburgh, Edinburgh, EH9 3JL, UK. (e-mail: y.fan@ed.ac.uk, chao.wang@ed.ac.uk, john.thompson@ed.ac.uk)}
\thanks{H. V. Poor is with Department of Electrical Engineering,
Princeton University, Princeton, NJ 08544 USA (email:
poor@princeton.edu)}}
%
%
%
\markboth{Submitted to IEEE Transactions on Wireless
Communications}{Fan \MakeLowercase{\textit{et al.}}: Recovering
Multiplexing Loss through Successive Relaying}
%


\maketitle

\begin{abstract}
In this paper, a transmission protocol is studied for a two relay
wireless network in which simple repetition coding is applied at the
relays. Information-theoretic achievable rates for this transmission
scheme are given, and a space-time V-BLAST signalling and detection
method that can approach them is developed. It is shown through the
diversity multiplexing tradeoff analysis that this transmission
scheme can recover the multiplexing loss of the half-duplex relay
network, while retaining some diversity gain. This scheme is also
compared with conventional transmission protocols that exploit only
the diversity of the network at the cost of a multiplexing loss. It
is shown that the new transmission protocol offers significant
performance advantages over conventional protocols, especially when
the interference between the two relays is sufficiently strong.
\end{abstract}


%

\section{Introduction}
%
%
%
%


\subsection{Background}

In the past few years, cooperative diversity protocols
\cite{1,3,4,5,6,7,8,9,10,11} have been studied intensively to
improve the diversity of relay networks. In most of the prior work,
a time-division-multiple-access (TDMA) half-duplex transmission is
assumed and the most popular transmission protocol (e.g.\cite{3})
can be described in two steps: In the first step, the source
broadcasts the information to all the relays. The relays process the
information and forward it to the destination (in either the same or
a different time slot) in the second step, while the source remains
\emph{silent}. The destination performs decoding based on the
message it received in both steps. We refer to this protocol as the
\emph{classic protocol} throughout the paper.

For digital relaying, where the relay decodes, re-encodes and
forwards the message, the simplest coding method is \emph{repetition
coding}\cite{3,4}, where the source and all the relays use the
\emph{same} codebook. This scheme can achieve full diversity and is
also practically implementable. Any capacity achieving AWGN channel
codes can be used to approach the performance limit of such schemes.
The disadvantage of this scheme is that it requires the relays to
transmit in orthogonal time slots in the second step in order for
the destination to combine effectively the relays' signals. This
will result in a significant \emph{multiplexing loss} compared with
direct transmission. Space-time codes, which were originally applied
in multiple-input multiple-output (MIMO) systems, have been
suggested for use in relay networks (e.g. \cite{4,anghel06}). Here
all the relays can transmit the signals simultaneously to the
destination in the second step and the multiplexing factor is
recovered to $1/2$. However, this still causes spectral inefficiency
for the high signal to noise ratio (SNR) region. In fact, the
network capacity in this scenario will become only \emph{half} of
the non-relay network capacity for high SNR, even assuming that the
message is always correctly decoded at the relays.

To fully recover the multiplexing loss, much more complicated
protocols and coding strategies have been proposed in \cite{7},
where \emph{new independent random codebooks} are used at the relays
to transmit the same information as they received from the source,
while the relays can adjust their listening times dynamically in the
first step. Those approaches, which are based on Shannon's random
coding theory, are currently theoretical and extremely difficult to
realize in reality. Practical coding design for relay networks often
follows a quite different approach from these theoretical
investigations (see \cite{janani04,stefanov05,hunter06,razaghi06}
for example).

Instead of using complicated coding schemes, a protocol using
\emph{repetition coded} relaying was proposed in \cite{5}(see also
\cite{wittneben03}) to avoid multiplexing loss for single relay
channels. In this protocol, denoted as protocol I in \cite{5}, the
source transmits a different message in the second time slot, so
that the destination sees a collision of messages from both the
relay and the source in the second time slot. Although multiplexing
loss is recovered due to the continuous transmission of the source,
diversity gain is lost due to the fact that the source transmission
in time slot two is not relayed to the destination.

\subsection{Contribution of the Paper}


In this paper, we study a transmission protocol based on protocol I
in \cite{5} for digital relaying. By adding an additional relay in
the network and making the two relays transmit in turn, we show that
multiplexing loss can be effectively recovered while
diversity/combining gain can still be obtained. Specifically, $L$
codewords can be transmitted in $(L+1)$ time slots with joint
decoding at the destination. Our analysis is based on two different
scenarios: (a) The instantaneous channel state information (CSI) is
known to the receiver and can be fed back to the transmitter; (b)
The instantaneous CSI is known to the receiver but is not available
to the transmitter. We make the following observations in this paper
for the proposed protocol:
\begin{itemize}
\item For scenario (a), we derive the achievable rates for this protocol when
repetition coding is assumed to be used at the relays. We show that
in certain scenarios, the capacity for the network becomes that for
a MIMO system with $L$ inputs and $L+1$ outputs and has a
multiplexing gain of $L/(L+1)$. Assuming that the relays correctly
decode the signal, we show that the proposed protocol offers
significant capacity performance advantages over the classic
protocol due to its improved multiplexing gain.
\item We also discuss the source-relay channel conditions and the
interference that arises between the relays. We derive the required
channel constraints for the optimal performance of such a protocol
as a function of SNR, as well as the achievable rates for different
channel conditions. We believe these analyses offer strong insights
for \emph{adaptive protocol design}, where relaying and direct
transmission can be combined. Based on our network models we show
that the proposed protocol, combined with the direct transmission
protocol, can also give a significant capacity performance advantage
over the classic protocol, especially when the two relays are
located close to each other.
\item We propose a practical low-rate feedback V-BLAST decoding algorithm
which approaches the theoretical achievable rates for a slow fading
environment.
\item For scenario (b), we analyze the diversity multiplexing tradeoff for such a
network when $L$ is large, conditioned on the signals being
correctly decoded at the relays. We show that in this scenario the
network mimics a multiple-input single-output (MISO) system with two
transmit and one receive antennas. This means it can offer a maximal
diversity gain of two with almost no multiplexing loss.
\end{itemize}

\subsection{Relations to Previous and Concurrent Work}

The idea for successive relaying first appeared in
\cite{Oechtering04}. This study was focused on amplify-and-forward
relaying and did not offer insight into the achievable rates and
diversity multiplexing tradeoff for such relaying methods. The
scheme has been further analyzed for amplify-and-forward relaying in
\cite{se2,yang061} and \cite{yang062}. In \cite{se2} capacity
analysis was performed assuming that the direct link is ignored.
Hence no diversity can be obtained at the destination. The analysis
in very recent papers \cite{yang061} and \cite{yang062} make the
assumption that the relays are isolated, when there are more than
one relay. Also the relay-to-relay link in \cite{yang061} and
\cite{yang062} acts in a different way from that in our work due to
a different relaying mode.


A further work \cite{se1} also analyzes the capacity for such
schemes when digital relaying is used. One major difference between
\cite{se1} and our work is that the direct link is ignored in
\cite{se1} while it is considered in this paper. Therefore the
analysis becomes different for the following reasons. First, the
scheme in \cite{se1} does not offer any \emph{cooperative diversity}
gain, which is a very important benefit that the relay can offer. We
will show in this paper that a diversity gain of 2 can be obtained
by considering the direct link. Secondly, this also results in very
different characteristics in terms of achievable rates and
signalling methods due to the additional interference and diversity
that the direct link introduces. In this paper we
\emph{specifically} analyze the network capacity under different
channel and interference constraints, which were not given in
\cite{se1}. Also the use of the V-BLAST decoder is unique to our
paper. Finally, we note that the capacity analysis discussed in our
paper in fact \emph{contains} the scenario in \cite{se1} as a
special case, i.e. the same capacity values as in \cite{se1} is
obtained on assuming that the channel coefficient for the direct
link is zero in our model. Therefore our analysis is more general,
and the adaptive protocols introduced here fit better in the context
of previous work on this topic \cite{3}.

\section{Protocol Design}

We assume a four-node network model, where one source, one
destination and two relays exist in the network. For simplicity, we
denote the source as $S$, the destination as $D$, and the two relays
as $R1$ and $R2$. We split the source transmission into different
frames, each containing $L$ codewords denoted as $s_l$. These $L$
codewords are transmitted continuously by the source, and are
decoded and forwarded by two relays successively in turn. Before
decoding $L$ codewords, the destination waits for $L+1$ transmission
time slots until all $L$ codewords are received, from both direct
link and the relay links. It then performs joint decoding of all $L$
codewords. The specific steps  for each transmission (reception)
time slot for every frame are described as follows:

\emph{Time slot 1}: $S$ transmits $s_1$. $R1$ listens to $s_1$ from
$S$. $R2$ remains silent. $D$ receives $s_1$.

\emph{Time slot 2}: $S$ transmits $s_2$. $R1$ decodes, re-encodes
and forwards $s_1$. $R2$ listens to $s_2$ from $S$ while being
interfered with by $s_1$ from $R1$. $D$ receives $s_1$ from $R1$ and
$s_2$ from $S$.

\emph{Time slot 3}: $S$ transmits $s_3$. $R2$ decodes, re-encodes
and forwards $s_2$. $R1$ listens to $s_3$ from $S$ while being
interfered with by $s_2$ from $R2$. $D$ receives $s_2$ from $R2$ and
$s_3$ from $S$. The progress repeats until \emph{Time slot L}.

\emph{Time slot L+1}: $R1$ (or $R2$) decodes, re-encodes and
forwards $s_L$. $D$ performs a joint decoding algorithm to decode
all $L$ codewords received from the $L+1$ transmission time slots.

The transmission schedule for the first three time slots for each
frame is shown in Fig. \ref{slot}. Compared with direct
transmission, the multiplexing ratio for this protocol is clearly
$L/(L+1)$, which approaches 1 for large frame lengths $L$. Unlike
protocol III in \cite{5}, the destination always receives two copies
of each codeword, from both the direct and relay link (a delayed
version). This implies that diversity gain can still be realized by
this protocol.

The major issue for this protocol to be effectively implemented is
to tackle the co-channel interference at the relays and the
destination. As described above, except for the first and last time
slot, the relays and the destination always observe collisions from
different transmitters (i.e. the source or the relays). Suppression
of the interference thus becomes a major problem. We will discuss
this problem further in the next two sections.

\section{Achievable Rates}

When the channel information can be fed back to the transmitter, one
can implement adaptive coding and modulation to achieve the system
performance limit. Thus capacity is a key measurement in this
scenario. We assume a slow, flat, block fading environment, where
the channel remains static for each message frame transmission (i.e.
$L+1$ time slots). Note that while this assumption is made for
presentation simplicity, the capacity analysis can also be applied
to a more relaxed flat block fading scenario, e.g. fast fading where
\emph{each channel coefficient changes for each time slot}. We also
assume that each transmitter transmits with equal power (i.e. no
power allocation or saving among the source and relays). We denote
the channel coefficient between node $a$ and $b$ by $h_{a,b}$, which
may contain path-loss, Rayleigh fading, and lognormal shadowing. For
simplicity, we denote the capacity function $\log _2 \left( {1 + x}
\right)$ by $C\left( x \right)$, in which the parameter $SNR$
denotes the ratio of signal power to the noise variance at the
receiver.

\subsection{Source-Relay Link}
\label{3a}

In order for the relays to decode the signals correctly, the source
transmission rate should be below the Shannon capacity of the
source-relay channels. We express this constraint as
\begin{equation}
R_i  \le C\left( {\left| {h_{S,r_i } } \right|^2 SNR } \right), 1
\le i \le L \label{srcap}
\end{equation}
where $r_i$ is the $i$th element in the $L$ dimensional relay index
vector
\begin{equation}
{\bf{r}} = {\left[ {\begin{array}{*{20}c}
   {R1} & {R2} & {R1} & {R2} & {R1} \cdots   \\
\end{array}} \right]},
\end{equation}
and $R_i$ denotes the achievable rate for $s_i$.

\subsection{Interference Cancellation Between Relays}
\label{3b} One major defect of the protocol is the interference
generated among the relays when one relay is listening to the
message from the source, while the other relay is transmitting the
message to the destination. This situation mimics a two user
Gaussian interference channel \cite{12}, where two transmitters (the
source and one of the relays) are transmitting messages each
intended for one of the two receivers (the other relay and the
destination). The optimal solution for this problem is still open.
We concern ourselves only with suppressing the interference at the
relays at this stage (interference suppression at the destination
will be left until all $L$ signals are transmitted). We give a very
simple decoding criterion for the relays: if the interference
between relays is stronger than the desired signal, we decode the
interference and subtract it from the received signals before
decoding the desired signal. Otherwise, we decode the signal
directly while treating the interference as Gaussian noise.

The achievable rate is therefore based on different channel
conditions between the source to relay and the relay to destination
links. For example, when $R1$ transmits $s_1$ while $R2$ is
receiving $s_2$, if $\left| {h_{R1,R2} } \right| \succ \left|
{h_{S,R2} } \right|$, $R2$ first decodes $s_1$, subtracts it (as the
interference), then decodes $s_2$ (as the desired signal).
Therefore, besides the rate constraint proposed in the previous
subsection, there will be an additional rate constraint for $s_1$ to
be correctly decoded at $R2$, which can be expressed as follows:
\begin{equation}
R_1 \le C\left( {\frac{{\left| {h_{R1,R2} } \right|^2 SNR}}{{1 +
\left| {h_{S,R2} } \right|^2 SNR}}} \right). \label{interrelay1}
\end{equation}
Otherwise if $s_2$ is decoded directly, treating $s_1$ as noise, the
achievable rate for $s_2$ is further constrained and can be
expressed as
\begin{equation}
R_2  \le C\left( {\frac{{\left| {h_{S,R2} } \right|^2 SNR}}{{1 +
\left| {h_{R1,R2} } \right|^2 SNR}}} \right). \label{rrcap}
\end{equation}
Note that this decoding criterion applies from the second time slot
to the $L$th time slot when transmitting each frame. In slot $i$,
inequality (\ref{interrelay1}) can be adapted to a constraint on $R
_{i-1}$ and inequality (\ref{rrcap}) can be adapted to a constraint
on $R _i$.

\subsection{Space-Time Processing at the Destination}

If the transmission rate is below the Shannon capacity proposed by
the previous two subsections, the relays can successfully decode and
retransmit the signals for all the $L+1$ time slots. The input
output channel relation for the relay network is equivalent to a
multiple access MIMO channel, which can be expressed as
\begin{equation}
{\bf{y}} = \sqrt{SNR} \underbrace {\left[ {\begin{array}{*{20}c}
   {h_{S,D} } & 0 & 0 & 0 & 0  \\
   {h_{r_1 ,D} } & {h_{S,D} } & 0 & 0 & 0  \\
   0 & {h_{r_2 ,D} } & {h_{S,D} } & 0 & 0  \\
   0 & 0 &  \ddots  &  \ddots  & 0  \\
   0 & 0 & 0 & {h_{r_{L - 1} ,D} } & {h_{S,D} }  \\
   0 & 0 & 0 & 0 & {h_{r_L ,D} }  \\
\end{array}} \right]}_{\bf{H}}{\bf{s}} + {\bf{n}},
\label{vstmimo}
\end{equation}
where $\bf{y}$ is the $(L+1) \times 1$ received signal vector,
$\bf{s}$ is the $L \times 1$ transmitted signal vector and
${\bf{n}}$ is an $(L+1) \times 1$ complex circular additive white
Gaussian noise vector at the destination. Unlike conventional
multiple access MIMO channels, the dimensions of $\bf{y}$, $\bf{s}$
and ${\bf{n}}$ are expanded in the time domain rather than the space
domain. However, the capacity region should be the same, which can
be expressed as follows \cite{13}:
\begin{eqnarray}
R_k  &\le& \log _2 \left( {\det \left( {{\bf{I}} + {\bf{h}}_k
{\bf{h}}_k^H SNR} \right)} \right),\label{rk}\\
R_{k_1 }  + R_{k_2 }  &\le& \log _2 \left( {\det \left( {{\bf{I}} +
SNR\left( {{\bf{h}}_{k_1 } {\bf{h}}_{k_1 }^H  + {\bf{h}}_{k_2 }
{\bf{h}}_{k_2 }^H } \right)} \right)} \right),\\
&...& \nonumber \\
\sum\limits_{k = 1}^L {R_k }  &\le& \log _2 \left( {\det \left(
{{\bf{I}} + {\bf{HH}}^H SNR} \right)} \right),\label{sumrk}
\end{eqnarray}
where ${\bf{h}}_k$ denotes the $k$th column of $\bf{H}$. As it is
extremely complicated to give an exact description for the rate
region of each signal when $L>2$, we will concentrate only on
inequalities (\ref{rk}) and (\ref{sumrk}) to give a sum capacity
upper bound for the network in the next subsection. However, as will
be shown later in the paper, this bound is extremely tight and is
achievable when a space-time V-BLAST algorithm is applied at the
destination to decode the signals in a slow fading scenario.

\subsection{Network Achievable Rates}

Combining the transmission rate constraints proposed by the previous
three subsections, we provide a way of calculating the network
capacity upper bound for the proposed protocol. First, we impose a
rate constraint $R_i$ for each transmitted codeword $s_i$. In the
first time slot (initialization), we write
\begin{equation}
R_{S,r_1 }  \le C\left( {\left| {h_{S,r_1 } } \right|^2 SNR}
\right). \label{ini}
\end{equation}

For $(i+1)$th time slot (for $1 \le i \le L-1$), we calculate the
rate constraints based on the decoding criterion at the relays. The
calculation can be written as a logical \texttt{if} statement as
follows:

\texttt{if} $h_{R1,R2}  \succ h_{S,r_{i + 1} }$,

\begin{eqnarray}
R_i  \le \min \left( {C\left( {\frac{{\left| {h_{R1,R2} } \right|^2
 SNR}}{{1 + \left| {h_{S,r_{i + 1} } } \right|^2 SNR }}} \right),R_{S,r_i }
,C\left( {\left| {h_{S,D} } \right|^2 SNR + \left| {h_{r_i ,D} }
\right|^2 } SNR \right)} \right), \nonumber \\ R_{S,r_{i + 1} } \le
C\left( {\left| {h_{S,r_{i + 1} } } \right|^2 } SNR
\right);\label{ryin}
\end{eqnarray}

\texttt{else}
\begin{eqnarray}
R_i  \le \min \left( {R_{S,r_i } ,C\left( {\left| {h_{S,D} }
\right|^2 SNR + \left| {h_{r_i ,D} } \right|^2 SNR} \right)}
\right),\nonumber \\ R_{S,r_{i + 1} }  \le C\left( {\frac{{\left|
{h_{S,r_{i + 1} } } \right|^2 SNR}}{{1 + \left| {h_{R1,R2} }
\right|^2 SNR}}} \right);\label{ryin1}
\end{eqnarray}

\texttt{end}.

Note that the term $C\left( {\left| {h_{S,D} } \right|^2 SNR +
\left| {h_{r_i ,D} } \right|^2 SNR} \right)$ represents the
constraint expressed by (\ref{rk}). The purpose of the \texttt{if}
statement is to select the decoding order at the relay and to decide
whether equation (\ref{interrelay1}) or (\ref{rrcap}) is the correct
constraint to apply.

In the $(L+1)$th time slot, we have
\begin{equation}
R_L  \le \min \left( {R_{S,r_L } ,C\left( {\left| {h_{S,D} }
\right|^2 SNR + \left| {h_{r_L ,D} } \right|^2 SNR} \right)}
\right).\label{ryin2}
\end{equation}

Combining these constraints with the sum capacity constraint
expressed by (\ref{sumrk}), an achievable rate per time slot can
then be written as
\begin{equation}
C_{ach}  = \frac{1}{L+1} \min \left( \mathop {\max }\limits_{R_1
\cdots R_L } \left\{ {\sum\limits_{i = 1}^L {R_i } } \right\},\log
_2 \left( {\det \left( {{\bf{I}} + {\bf{HH}}^H SNR} \right)} \right)
\right). \label{capbound}
\end{equation}
The first term in the $\min$ function comes from the calculation
described above, the second one comes from equation (\ref{sumrk}).

\subsection{Interference Free Transmission}
\label{intf} From the above discussion of the proposed protocol, it
is clear that the interference between relays is one major and
obvious factor that can significantly degrade the network capacity
performance. However, it has been shown that for a Gaussian
interference network, if the interference is sufficiently strong,
the network can perform the same as an interference free network
\cite{14}. Specifically, for the scenario discussed in our model, if
the interference between relays (i.e. the value of $\left| h_{R1,R2}
\right|$) is so large that the following inequality holds
\begin{equation}
 \frac{{\left| {h_{R1,R2} } \right|^2 SNR}}{{1 + \left| {h_{S,r_{i +
1} } } \right|^2 SNR}} \ge \min \left( {\left| {h_{S,r_i } }
\right|^2 SNR,\left( \left| {h_{S,D} } \right|^2 + \left| {h_{r_i
,D} } \right|^2 \right) SNR}\right), i = 1 \cdots L, \label{int}
\end{equation}
then the relay can always correctly decode the interference and
subtract it before decoding the desired message, without affecting
the overall network capacity. In this situation, the capacity
analysis for the $i$th ($1 \le i \le L$) transmitted signal as
expressed by (\ref{ryin})-(\ref{ryin2}) can be simplified to
\begin{equation}
R_i \le \min \left( C\left( {\left| {h_{S,r_{i} } } \right|^2 SNR}
\right), C\left( {\left(\left| {h_{S,D} } \right|^2  + \left|
{h_{r_i ,D} } \right|^2 \right)SNR} \right) \right).
\label{interfree}
\end{equation}
It is obvious that the rate bounds provided by (\ref{interfree}) are
significantly larger than those provided by
(\ref{ryin})-(\ref{ryin2}).

From the above capacity analysis, it can also be seen that the
quality of the source to relay link (i.e. $h_{S,r_{i} }$) is also an
important factor that may constrain the network capacity. This has
also been justified and discussed in many papers (e.g.
\cite{3,4,10,7,9}). Similar to this previous work, we suggest that
$h_{S,r_{i} }$ should be compared with $h_{S,D}$ or $h_{r_i ,D}$
before deciding to relay or not. For the interference free scenario
discussed here, the constraint becomes
\begin{equation}
\left| {h_{S,r_{i} } } \right|^2 \ge \left| {h_{S,D} } \right|^2  +
\left| {h_{r_i ,D} } \right|^2, 1 \le i \le L. \label{sr}
\end{equation}
The capacity expressed by (\ref{capbound}) can be simplified to
\begin{equation}
C_{ach}  = \frac{1}{L+1}\min \left( {\sum\limits_{i = 1}^L {C \left(
\left( {\left| {h_{S,D} } \right|^2  + \left| {h_{r_i ,D} }
\right|^2 } \right)SNR \right)} ,\log _2 \left( {\det \left(
{{\bf{I}} + {\bf{HH}}^H SNR} \right)} \right)} \right).
\label{capbound1}
\end{equation}
By Jensen's inequality \cite{14} it is clear that
\begin{equation}
\sum\limits_{i = 1}^L {C \left( \left( {\left| {h_{S,D} } \right|^2
+ \left| {h_{r_i ,D} } \right|^2 } \right) SNR \right) } \ge \log _2
\left( {\det \left( {{\bf{I}} + {\bf{HH}}^H SNR} \right)} \right).
\end{equation}
Therefore the rate is equal to the MIMO channel capacity equation
with a multiplexing scaling factor:
\begin{equation}
C_{ach}  = \frac{1}{L+1} \log _2 \left( {\det \left( {{\bf{I}} +
{\bf{HH}}^H SNR} \right)} \right). \label{mimo}
\end{equation}
This result shows that the proposed protocol can offer the best
capacity performance conditioned on (\ref{int}) and (\ref{sr}),
which guarantees that the relays will correctly decode the message
without affecting the network capacity. To summarize, we have the
following theorem.

\newtheorem{theorem}{Theorem}
\begin{theorem}
Conditioned on (\ref{int}) and (\ref{sr}), the capacity for the
successive relaying scheme can be expressed as
\begin{equation}
C_{ach}  = \frac{1}{L+1} \log _2 \left( {\det \left( {{\bf{I}} +
{\bf{HH}}^H SNR} \right)} \right) \nonumber
\end{equation}
where $\bf{H}$ denotes the channel matrix in (\ref{vstmimo}).
\end{theorem}

It should be noted that this high interference scenario (i.e.
condition (\ref{int})) is \emph{not uncommon} in reality. A
practical example is when the two relays (e.g. mobiles) are located
close to each other. If the routing techniques are developed to
choose these relays, the capacity performance can be significantly
improved by applying the proposed protocol. To satisfy condition
(\ref{sr}), an adaptive protocol can be developed from the proposed
protocol to guarantee that the relays are used only when (\ref{sr})
holds, otherwise direct transmission is assumed. However, for a
large dense network of relays, it is even not difficult to find two
relays satisfying both (\ref{int}) and (\ref{sr}). A simple example
is a fixed relay network scenario \cite{15}, where the source to
relay links are often assumed to be significantly better than the
corresponding relay to destination links and the direct link.
Therefore both (\ref{int}) and (\ref{sr}) can be met by choosing the
two nearby fixed relays. Furthermore, studies have shown that for a
large relay network where many relays exist, choosing the best one
or few relays will be preferable to using all the relays in many
situations (e.g. \cite{8,16,17,19}). Therefore it is possible that
the proposed relay protocol can be combined with relay selection
techniques to achieve an even higher capacity gain over the classic
multi-cast relay protocol, especially for high SNR conditions.

\subsection{The V-BLAST Algorithm}


In this section we apply the low-rate feedback V-BLAST minimum mean
squared error (MMSE) algorithm for detecting the signals at the
destination. The V-BLAST algorithm was initially designed for
spatial multiplexing MIMO systems \cite{blast}. For a system with
$M$ transmit and $N$ receive antennas, the message at the
transmitter is multiplexed into $M$ different signal streams, each
independently encoded and transmitted to the receiver. The receiver
uses $N$ antennas to detect and decode each signal stream by a
V-BLAST MMSE detector \cite{vblast}. The V-BLAST MMSE detection
consists of $M$ iterations, each aimed at decoding one signal
stream. For each iteration, the receiver applies the MMSE algorithm
to detect and decode the \emph{strongest} signal while treating the
other signals as interference, then subtracts it from the received
signal vector. The detection continues until all $M$ signal streams
are decoded. The Shannon capacity of this system can be achieved if
we assume that each signal is correctly decoded\cite{cdma}:
\begin{eqnarray}
C = \log _2 \det \left( {{\bf{I}} + {\bf{HH}}^H } SNR \right) =
\sum\limits_{i = 1}^M {\log _2 } \left( {1 + SINR _i } \right),
\label{4}
\end{eqnarray}
where $SINR _i$ is the output signal to interference plus noise
ratio (SINR) for signal $s_i$ in the V-BLAST detector. In order for
each signal to be correctly decoded, a low-rate feedback channel can
be used to feed the value of $SINR_i$ back to the transmitter.
Adaptive modulation and coding should be applied to make the
transmission rate for $s _i$ lower than ${\log _2 } \left( {1 + SINR
_i } \right)$.

Unlike traditional MIMO systems, when we apply this V-BLAST MMSE
detector at the destination for the proposed protocol, each signal
stream is independently encoded along the \emph{time dimension}
rather than the space dimension. When considering the rate bound
$R_i$, the same analysis should be made as in Section III. The
initialization step is the same as (\ref{ini}). For the $(i+1)$th
time slot (for $1 \le i \le L-1$), based on the same interference
cancellation criterion as in Section III, the rate calculation can
be performed as follows:

\texttt{if} $h_{R1,R2}  \succ h_{S,r_{i + 1} }$

\begin{eqnarray}
R_i  \le \min \left( {C\left( {\frac{{\left| {h_{R1,R2} } \right|^2
SNR}}{{1 + \left| {h_{S,r_{i + 1} } } \right|^2 SNR}}}
\right),R_{S,r_i } ,\log _2 \left( {1 + SINR_{r_i } } \right)}
\right),\nonumber \\ R_{S,r_{i + 1} } \le C\left( {\left| {h_{S,r_{i
+ 1} } } \right|^2 SNR } \right);
\end{eqnarray}

\texttt{else}

\begin{eqnarray}
R_i  \le \min \left( {R_{S,r_i } ,\log _2 \left( {1 + SINR_{r_i } }
\right)} \right),R_{S,r_{i + 1} }  \le C\left( {\frac{{\left|
{h_{S,r_{i + 1} } } \right|^2 SNR}}{{1 + \left| {h_{R1,R2} }
\right|^2 SNR}}} \right);
\end{eqnarray}

\texttt{end}.

In the $(L+1)$th time slot, we have
\begin{equation}
R_L  \le \min \left( {R_{S,r_L } ,\log _2 \left( {1 + SINR_{r_L } }
\right)} \right).
\end{equation}
The $SINR_{r_i }$ denotes the SINR for $s _i$, which is decoded,
encoded and forwarded by relay $r_i$. The network capacity is
therefore
\begin{equation}
C_{ach_{BLAST}}  = \frac{1}{L+1} \mathop {\max }\limits_{R_1  \cdots
R_L } \left\{ {\sum\limits_{i = 1}^L {R_i } } \right\}.
\label{capblast}
\end{equation}

The condition for interference free transmission discussed in
Section \ref{intf} can be expressed as
\begin{equation}
C\left( {\frac{{\left| {h_{R1,R2} } \right|^2 SNR}}{{1 + \left|
{h_{S,r_{i + 1} } } \right|^2 SNR}}} \right) \ge \min \left(
{C\left( {\left| {h_{S,r_i } } \right|^2 SNR} \right),\log _2 \left(
{1 + SINR_{r_i } } \right)} \right). \label{inblast}
\end{equation}
The rate for the $i$th ($1 \le i \le L$) signal under this condition
can be expressed as
\begin{equation}
R _i \le \min \left( {C\left( {\left| {h_{S,r_i } } \right|^2 SNR}
\right),\log _2 \left( {1 + SINR_{r_i } } \right)} \right).
\end{equation}
Similar to the discussion in Section \ref{intf}, we can further
apply adaptive protocols or make relay selections in the network to
enhance the source to relay links; i.e.,
\begin{equation}
C\left( {\left| {h_{S,r_i } } \right|^2 }SNR \right) \ge \log _2
\left( {1 + SINR_{r_i } } \right), \label{vblastmmse}
\end{equation}
and it is clear from (\ref{4}) that (\ref{capblast}) equals
(\ref{mimo}) under conditions (\ref{vblastmmse}) and
(\ref{inblast}). This implies that the V-BLAST algorithm can achieve
rate (\ref{mimo}) for the protocol if the interference channel
between relays and source to relay channels are sufficiently strong.

It can be seen that the conditions in (\ref{inblast}) and
(\ref{vblastmmse}) have a higher probability of being fulfilled than
those in (\ref{int}) and (\ref{sr}) due to the following
observation:
\begin{equation}
SINR_{r_i} \le \left( \left| {h_{S,D} } \right|^2  + \left| {h_{r_i
,D} } \right|^2 \right) SNR.
\end{equation}
This further implies that the conditions in (\ref{inblast}) and
(\ref{vblastmmse}) are better suited to assist the VBLAST algorithm
to achieve the rate in (\ref{mimo}), than those in (\ref{int}) and
(\ref{sr}). We note that in practice these conditions also imply a
signalling overhead among the source, relays and destination in
order to obtain the required SINR information. Furthermore, we note
that V-BLAST might be applied only to a slow fading scenario in
which the channel remains unchanged at least in every $L+1$
transmission time slots. This is due to the fact that \emph{SINR}
has to be fed back to the transmitters \emph{before} the source
starts transmitting at the beginning of the \emph{L+1} time slots.

\subsection{Comparison with Classic Protocols}

\subsubsection{Classic Protocol I}

The first classic protocol was presented by Laneman and Wornell
\cite{4}, where each message transmission is divided into three time
slots. In the first time slot, the source broadcasts the message to
the two relays and the destination. In the next two time slots, each
relay retransmits the message to the destination in turn after
decoding and re-encoding it by repetition coding. The destination
combines the signals it receives in the three time slots. The
network capacity for this protocol can be written as:
\begin{eqnarray}
C = \frac{1}{3} \times \min \Bigl( && C\left( {\left| {h_{S,R1} }
\right|^2 } SNR \right),C\left( {\left| {h_{S,R2} } \right|^2 } SNR
\right),\nonumber\\ & & C\left( \left( {\left| {h_{S,D} } \right|^2
+ \left| {h_{R1,D} } \right|^2  + \left| {h_{R2,D} } \right|^2 }
\right) SNR \right) \Bigr), \label{capcla1}
\end{eqnarray}
where the term $\frac{1}{3}$ denotes the multiplexing loss compared
with direct transmission.

\subsubsection{Classic Protocol II}

A simple improvement of Classic Protocol I is to apply distributed
Alamouti codes at the relays \cite{6}. The system uses four time
slots to transmit two signals. In the first two time slots the
source broadcasts $s _1$ and $s _2$ to both the relays and the
destination. In the next two time slots $R1$ transmits $\left[ {s_1
, - s_2^* } \right]$ and $R2$ transmits $\left[ {s_2 ,s_1^* }
\right]$. The destination uses maximal ratio combining to combine
the signals received from all four time slots in order to detect and
decode them. The capacity achieved by this protocol can be written
as
\begin{eqnarray}
C = \frac{1}{2} \times \min \Bigl( & & C\left( {\left| {h_{S,R1} }
\right|^2 } SNR \right),C\left( {\left| {h_{S,R2} } \right|^2 } SNR
\right),\nonumber\\ & & C\left( \left( {\left| {h_{S,D} } \right|^2
+ \left| {h_{R1,D} } \right|^2  + \left| {h_{R2,D} } \right|^2 }
\right) SNR \right) \Bigr). \label{capcla2}
\end{eqnarray}
It is clear that (\ref{capcla2}) outperforms (\ref{capcla1}) as it
has the same diversity gain but reduced multiplexing loss compared
with direct transmission.

In practice, both protocols can be combined with relay selection or
adaptive relaying protocols to make sure that
\begin{equation}
\min \left( {C\left( {\left| {h_{S,R1} } \right|^2 } SNR
\right),C\left( {\left| {h_{S,R2} } \right|^2 } SNR \right)} \right)
\ge C\left( \left( {\left| {h_{S,D} } \right|^2  + \left| {h_{R1,D}
} \right|^2 + \left| {h_{R2,D} } \right|^2 } \right) SNR \right)
\label{clasr}
\end{equation}
when relaying is used. The network under this condition can achieve
the best capacity performance (i.e. the third term in
(\ref{capcla1}) and (\ref{capcla2})). This result clearly mimics the
performance of a $3 \times 1$ single-input multiple-output (SIMO) or
multiple-input single-output (MISO) system.

\subsubsection{Performance Comparison}

It can be seen that if the two relays are close to each other so
that (\ref{int}) holds, then condition (\ref{sr}) is more likely to
hold than (\ref{clasr}). This implies that the best capacity
(\ref{mimo}) for the proposed protocol can be achieved with a higher
probability than that for the classic protocols. We now simply
compare the best capacities can be achieved by both proposed
protocol and Classic Protocol II:
\begin{equation}
G \buildrel \Delta \over = \frac{{\frac{1}{{L + 1}}{\rm E}\left[
{\log _2 \left( {\det \left( {{\bf{I}} + {\bf{HH}}^H SNR} \right)}
\right)} \right]}}{{0.5 \times {\rm E}\left[ {C\left( \left( {\left|
{h_{S,D} } \right|^2  + \left| {h_{R1,D} } \right|^2  + \left|
{h_{R2,D} } \right|^2 } \right) SNR \right)} \right]}},
\end{equation}
where ${\rm E}\left[ \bullet \right]$ denotes the expectation and we
assume $\left\{ {h_{a,b} } \right\}$ is a set of identically,
independent distributed (i.i.d), complex, zero mean Gaussian random
variables with unit variances. $G$ is plotted as a function of $SNR$
in Fig. \ref{succgain} for different values of $L$. It is clear that
the capacity gain increases as the value of SNR increases. Larger
values of $L$ lead to reduced multiplexing loss and offer higher
capacity gains.

\section{Diversity multiplexing tradeoff}

When the instantaneous CSI is not known to the transmitter, outages
will occur. In this scenario diversity-multiplexing tradeoff is a
powerful tool to measure the balance between the rates and error
probability. In this section we study further the diversity
multiplexing tradeoff \cite{DMT} for such a protocol. For simplicity
our analysis is based on the assumption that the signals are
correctly decoded at the relays. We note that this analysis can
provide insights on the best possible performance this scheme can
offer. We summarize our results in the following.
\begin{theorem}
Define the diversity gain $d$ and multiplexing gain $r$ as those in
\cite{DMT}. Conditioned on the relays correctly decoding the signals
(i.e, (\ref{int}) and (\ref{sr})\footnote{Note that the probability
that (\ref{int}) holds decreases as the SNR increases. Therefore,
the theorem offers an upper bound on the performance of such a
system at high SNR}), the diversity multiplexing tradeoff for the
successive relaying scheme in a slow fading scenario, where the
channel coefficients remain the same for $L+1$ time slots, can be
expressed as:
\begin{equation}
d(r)=2\left(1-\frac{L+1}{L}r\right)^{+}. \label{eq:DMT}
\end{equation}
\end{theorem}

\begin{proof}
See Appendix.
\end{proof}

As predicted in the previous section, we can see from this theorem
that a maximal diversity gain of 2 can be obtained, while the
multiplexing gain can be recovered to nearly 1 for large $L$. This
will offer a significant advantage in terms of spectral efficiency,
which will be shown through simulations in the next section. Table
\ref{table2} compares the maximal diversity and multiplexing gains
between the successive relaying protocol and the classic protocols
in a slow fading scenario. Note that for a faster fading scenario
where the channel coefficient changes in every transmission time
slot, the same theorem still holds if the signal transmitted in each
time slot is independently encoded. Furthermore, we note that it is
possible to obtain the same diversity-multiplexing tradeoff
performance if proper adaptive protocols similar to those in
\cite{3} are used to consider the conditions of the source to relay
links \cite{chao}.

\section{Simulation results }


In this section we make further comparison of the above protocols
for different network geometries in terms of achievable rates. We
compare only Classic Protocol II with the proposed protocol. As
mentioned previously, to achieve better capacity performance in
practice, the classic protocols should be combined with adaptive
protocols so that relaying is applied only if the source to relay
channels are good. There are a number of ways to enable adaptive
protocols. Three examples are to base adaptation on one of the
following conditions: (a) $\min \left( {\left| {h_{S,R1} }
\right|,\left| {h_{S,R2} } \right|} \right) \ge \left| {h_{S,D} }
\right|$, i.e. the source to relay link is better than the direct
link; (b) condition (\ref{sr}) holds; or (c) condition (\ref{clasr})
holds. Although (b) and (c) fits better with the analysis in this
paper, condition (a) is the \emph{simplest} since it does not
require knowledge of the relay to destination links. In the
following we will only adopt (a) in the simulations. I.e, if
condition (a) is not met, the system will use direct transmission.
Similar results would be obtained if condition (b) or (c) were to be
adopted instead.

Our simulations are based on three network geometries: cases I, II
and III, which are shown in Fig.\ref{model}. We assume that each
$h_{a,b}$ contains Rayleigh fading, pathloss and independent
lognormal shadowing terms. These terms can be written as $h_{a,b}  =
{v_{a,b}} \sqrt {{x_{a,b}}^{ - \gamma } 10^{\zeta _k /10}}$, where
$\left\{ {v_{a,b} } \right\}$ is a set of i.i.d. complex Gaussian
random variables with unit variances, and $x_{a,b}$ is the distance
between the nodes $a$ and $b$. The scalar $\gamma$ denotes the path
loss exponent (in this paper it is always set to 4). The lognormal
shadowing term $\zeta_k$ is a random variable drawn from a normal
distribution with a mean of $0$ dB and a standard deviation
$\delta=8$ (dB). We assume that the distance between the source and
destination is normalized to unit distance. In case I, the distances
between the source to relays and relays to destination are all
normalized, so the distance between the two relays is therefore
$\sqrt{3}$. In case II, the distance between relays is normalized,
while the distance between the source to relays and relays to
destinations is $1/\sqrt{2}$. In case III, the relays are located in
the middle region between the source and destination, so that the
distance between the source and relays is $1/2$ while the distance
between the relays is negligible compared with the source to relays
links. For the proposed protocol, these three cases represent a
meaningful tradeoff between the strength of source to relay channels
and the interference channel between the two relays.

We assume $L=7$ in the simulation, and the performance for the
proposed protocol will certainly increase as $L$ increases.
Fig.\ref{sim} shows the achievable rates for the proposed protocols
(ach rate), the capacity achieved by V-BLAST MMSE detection
(VBLAST), the classic protocols (classic) and direct transmission
(direct), all averaged over $10,000$ channel realizations. It can be
clearly seen from all three figures that the V-BLAST algorithm
approaches the capacity bounds for the protocol proposed in this
paper.

Both Fig.\ref{case1} and Fig. \ref{case2} imply that it is generally
not helpful to implement relaying protocols when the source to relay
link is about the same quality as the source to destination (direct)
link, as the link gain due to relaying is small in this case.
However, the proposed protocol still offers a performance gain over
direct transmission for both the high and low SNR regions in these
cases. Compared with case I and II, in case III the source to relay
links are much stronger, and the relays become close to each other
so that the interference is sufficiently strong to allow
interference free transmission, as discussed in Sections III and IV.
It can be clearly seen in Fig. \ref{case3} that the proposed
protocol gives a significant performance advantage over direct
transmission for both low and high SNR regions due to its combining
gain and negligible multiplexing loss. The classic protocol still
performs worse than direct transmission due to its significant
multiplexing loss compared with direct transmission, although its
performance gain over direct transmission for the low SNR region is
improved.


%

%

%

\section{Conclusions}

In this paper we have analyzed successive relaying protocol. Our
analysis shows that this protocol can maintain combining/diversity
gain while recovering the multiplexing loss associated with the
classic protocol. We have proposed the use of a low complexity
V-BLAST detection algorithm to help implement this protocol
effectively. From the simulation study based on different
geometries, we can draw two main conclusions: (a) For both the
proposed and classic protocols, the network capacity increases when
the source-relay link becomes stronger; (b) in this scenario, while
the classic protocol still loses its performance advantage for the
high SNR region, the proposed protocol can give significant
performance advantages for both the low and high SNR regions.

Note that one very important factor that impairs the capacity
performance of the proposed protocol is interference between the two
relays. Our capacity analysis does not offer the optimal capacity
results for this protocol because the optimal method of suppressing
the interference between the relays is not known in general. For the
adaptive protocol discussed in the paper, it is also worthwhile to
develop alterative forms of the protocol that explicitly account for
the impact of interference between relays on the network capacity.
Also it should be interesting to extend the analysis into a more
than two relay scenario. These are interesting topics for future
work.


%
%
\appendix[Proof of \emph{Theorem 1}]


As mentioned in Section III.E, conditioned on the event that the
relays correctly decode the message, the successive relaying
protocol mimics a multiple access MIMO channel (\ref{vstmimo}) with
a capacity constraints (\ref{rk}) - (\ref{sumrk}). For each
constraint there is a probability of not meeting it. The probability
of outage is the highest among all these probabilities. Therefore
there are $(2^L-1)$ diversity-multiplexing tradeoffs for all those
conditions and the lowest curve within the range of multiplexing
gain is the optimal tradeoff curve for the system \cite{dmtac}. To
characterize the diversity-multiplexing tradeoff achieved by each
constraint, we consider an $(m+1) \times m$ MIMO channel matrix
${\bf{H}}_m$ in the same form as in (\ref{vstmimo}). Define $v_0$ as
the exponential order \cite{7} of $1/|h_{S,D}|^2$ and $v_k$ as the
exponential order of $1/|h_{r_k,D}|^2$. Furthermore, Let
$\textbf{M}_{m+1}=\textbf{I}+\frac{1}{2} \Sigma_S \Sigma_n^{-1}$,
where $\Sigma_S$ and $\Sigma_n$ denote the covariance matrices of
the observed signal and noise components \emph{at the receiver},
respectively. We assume that each source message $s_i$ is chosen
from a Gaussian random codebook of codeword length $l$. When $m=1$,
the upper bound on the ML conditional pair-wise error probability
(PEP) can be calculated by
\begin{eqnarray}
P_{PE\left| {v_0 ,v_1 } \right.}&  \le & \det \left( {\textbf{I} +
\frac{1}{2} \Sigma_{S_1} \Sigma_{n}^{ - 1} }\right)^{-l} \nonumber\\
&  = & {}\left( {1 + \frac{1}{2}\rho \left| {h_{S,D} } \right|^2  +
\frac{1}{2}\rho \left| {h_{r_1,D} } \right|^2 } \right)^{ - l}
\nonumber \\ & \dot = & {} \rho^{-l ( \max \{1-v_0,1-v_1\})^+}
\end{eqnarray}
where $\dot =$ denotes the exponential equality \cite{DMT} and $SNR$
is replaced by $\rho$ for notational simplicity. We assume each
$s_i$ is transmitted with data rate $R$ bits in each transmission
time slot. Since the successive relaying protocol uses $(L+1)$ time
slots to transmit $L$ different symbols, the average transmission
rate is $\overline{R}=\frac{L}{L+1}R$. On assuming that average
transmission rate changes as $\overline{R}=r\log \rho$ with respect
to $\rho$, then it is easy to see $R=\frac{L+1}{L}r\log \rho$.
Therefore, we have a total of $\rho ^{\frac{L+1}{L}rl}$ codewords.
Thus, the error probability can be bounded by
\begin{equation}
P_{E|v_0,v_1} \dot \leq \rho^{-l((\max
\{1-v_0,1-v_1\})^{+}-\frac{L+1}{L}r)}.
\end{equation}
Next, we want to find the set in which the outage event always
dominates the error probability performance. The analysis regarding
this is similar to that in \cite{7} and is thus omitted here. This
set is given by
\begin{equation} \label{eq:oute}
O^+=\left\{ {(v_0,v_1) \in R^{2+}}
 \mathrel{\left | {\vphantom {(v_0,v_1) \in R^{2+}| (\max \{1-v_0,1-v_1\})^+ \leq
\frac{L+1}{L}r}}
 \right. \kern-\nulldelimiterspace}
 {(\max \{1-v_0,1-v_1\})^+ \leq
\frac{L+1}{L}r} \right\}.
\end{equation}
Then, for any error event which belongs to the non-outage set, we
can choose $l$ to make its probability sufficiently small to ensure
that the error performance is dominated by the outage probability,
which can be expressed as $\rho^{-d_o(r)}$ for
$d_o(r)=\mathop{\inf}\limits_{(v_0,v_1) \in O^{+}} (v_0+v_1)$. Now
using (\ref{eq:oute}), $d_o(r)$ can be calculated as
\begin{equation} \label{eq:DMT1}
d_o(r)=2\left(1-\frac{L+1}{L}r\right)^+
\end{equation}
which represents the diversity-multiplexing tradeoff in the case
$m=1$. When $m \geq 2$, the analysis of the determinant of
$\textbf{M}_{m+1}$ can be conducted in a way similar to that in
\cite{yang062}; so we omit the specific calculation due to limited
space. Define $D_k:=\det(\textbf{M}_{(k)})$, where
$\textbf{M}_{(k)}$ denotes a $k \times k$ sub-matrix formed by the
first $k$ rows and $k$ columns from the upper left-most corner of
$\textbf{M}$. The coefficients of $D_{m+1}$ can be calculated
recursively as
\begin{displaymath}
D_{m+1}\left(\frac{1}{2}\rho |h_{S,D}|^2\right) =
\left(\frac{1}{2}\rho |h_{S,D}|^2\right)^m + \prod \limits_{j =
1}^{m} {\left(1 + \frac{1}{2}\rho |h_{r_j ,D}|^2 \right)} +
P\left(\frac{1}{2}\rho |h_{S,D}|^2\right)
\end{displaymath}
where $P(\frac{1}{2}\rho |h_{S,D}|^2)$ is a polynomial in
$\frac{1}{2}\rho |h_{S,D}|^2$ and is always nonnegative. Thus, we
have
\begin{equation}
D_{m+1} \geq \left(\frac{1}{2}\rho |h_{S,D}|^2\right)^m + \prod
\limits_{k = 1}^{m} {\left(1 + \frac{1}{2}\rho |h_{r_k ,D}|^2
\right)}.
\end{equation}
Since we assume a slow fading environment, $v_1=v_3=\dots$ and
$v_2=v_4=\dots$. On setting $v=\max \{v_1,v_2\}$, it can be seen
that
\begin{equation} \label{eq:detvalue}
\det ({\bf I} + \Sigma _{S_m} \Sigma _{n}^{ - 1} )\dot \geq
\rho^{\max \{m(1-v_0)^+,m(1-v)^+\}}.
\end{equation}
If we define $\det ({\bf I} + \Sigma _{S_m} \Sigma _{n}^{ - 1} )
\dot = \rho^{f(v_0,v_1,v_2)}$ and
\begin{equation} \label{eq:gequ}
\rho^{\max \{m(1-v_0)^+,m(1-v)^+\}} \dot = \rho^{g(v_0,v_1,v_2)},
\end{equation}
then we have
\begin{equation} \label{eq:flg}
f(v_0,v_1,v_2) \dot \geq g(v_0,v_1,v_2), ~~~~~\forall
(v_0,v_1,v_2)\in R^{3+}.
\end{equation}
Similarly to the analysis for $m=1$, $O_f^+$ should be defined as
\begin{equation} \
O_f^+=\left\{ {(v_0,v_1,v_2) \in R^{3+}}
 \mathrel{\left | {\vphantom {(v_0,v_1,v_2) \in R^{3+} f(v_0,v_1,v_2) \leq
\frac{L+1}{L}mr}}
 \right. \kern-\nulldelimiterspace}
 {f(v_0,v_1,v_2) \leq
\frac{L+1}{L}mr} \right\}
\end{equation}
where $m$ denotes that $m$ symbols are transmitted and the
equivalent data rate $R=\frac{L+1}{L}mr\log \rho$. We define
\begin{equation} \label{eq:gdiver}
O_g^+=\left\{ {(v_0,v_1,v_2) \in R^{3+}}
 \mathrel{\left | {\vphantom {(v_0,v_1,v_2) \in R^{3+} g(v_0,v_1,v_2) \leq
\frac{L+1}{L}mr}}
 \right. \kern-\nulldelimiterspace}
 {f(v_0,v_1,v_2) \leq
\frac{L+1}{L}mr} \right\}.
\end{equation}
Because of (\ref{eq:flg}), it can be seen that $O_f^+ \subseteq
O_g^+$. Therefore
\[
\mathop{\inf}\limits_{(v_0,v_1,v_2) \in O_f^{+}} (v_0+v_1+v_2) \geq
\mathop{\inf}\limits_{(v_0,v_1,v_2) \in O_g^{+}} (v_0+v_1+v_2),
\]
which means that the diversity gain calculated from $O_f^+$ is
always larger than that calculated from $O_g^+$.

From (\ref{eq:gequ}) and (\ref{eq:gdiver}), it is not difficult to
show that
\begin{equation} \label{eq:DMT2}
\mathop{\inf}\limits_{(v_0,v_1,v_2) \in O_g^{+}} (v_0+v_1+v_2) \ge
2\left(1-\frac{L+1}{L}r \right)^+.
\end{equation}
Comparing (\ref{eq:DMT1}) and (\ref{eq:DMT2}), we can see that the
diversity gain achieved by a multiple access MIMO channel with
channel matrix $\textbf{H}_m$ $(m>1)$is always larger than that for
$\textbf{H}_1$.

Now we consider the product of the determinants of $n$ matrices
$\prod\limits_{i = 1}^n {\det({\bf I} + \frac{1}{2} \Sigma _{S_{m_i}
} \Sigma _{n}^{ - 1} )}$, which is related to all other rate
constraints from (\ref{rk}) - (\ref{sumrk}). Using
(\ref{eq:detvalue}), it is easy to obtain
\[
\prod\limits_{i = 1}^n {\det\left({\bf I} + \frac{1}{2} \Sigma
_{S_{m_i} } \Sigma _{n}^{ - 1} \right)} \dot \geq \rho^{\max
\{(\sum_{i=1}^n m_i)(1-v_0)^+,(\sum_{i=1}^n m_i)(1-v)^+\}}
\]
Define $\rho^{f_n(v_0,v_1,v_2)} \dot = \prod\limits_{i = 1}^n
{\det({\bf I} + \frac{1}{2} \Sigma _{S_{m_i} } \Sigma _{n}^{ - 1}
)}$ and $\rho^{g_n(v_0,v_1,v_2)} \dot = \rho^{\max \{(\sum_{i=1}^n
m_i)(1-v_0)^+,(\sum_{i=1}^n m_i)(1-v)^+\}}$. It can be seen that
\begin{equation}
f_n(v_0,v_1,v_2) \dot \geq g_n(v_0,v_1,v_2), ~~~~~\forall
(v_0,v_1,v_2)\in R^{3+}.
\end{equation}
Similarly, applying $O_{g_n}^+=\{(v_0,v_1,v_2) \in R^{3+}|
g(v_0,v_1,v_2) \leq \frac{L+1}{L}(\sum_{i=1}^n m_i)r \}$, we have
that
\[
\mathop{\inf}\limits_{v_0,v_1,v_2 \in O_{f_n}^{+}} (v_0+v_1+v_2)
\geq 2\left(1-\frac{L+1}{L}r\right)^{+}.
\]
The determinant of the matrix $({\bf I} + \frac{1}{2} \Sigma
_{S}\Sigma _{n }^{ - 1})$ can always be decomposed into the product
of the determinants of several submatrices $({\bf I} + \frac{1}{2}
\Sigma _{S_{mi}}\Sigma _{n}^{ - 1})$. Therefore the error exponent
is always larger than or equal to $2(1-\frac{L+1}{L}r)^{+}$ and the
proof is complete.





\renewcommand{\baselinestretch}{1}
%

%

\begin{figure}[p!]
\subfigure[Time slot 1.]{\includegraphics[width=2.5in]{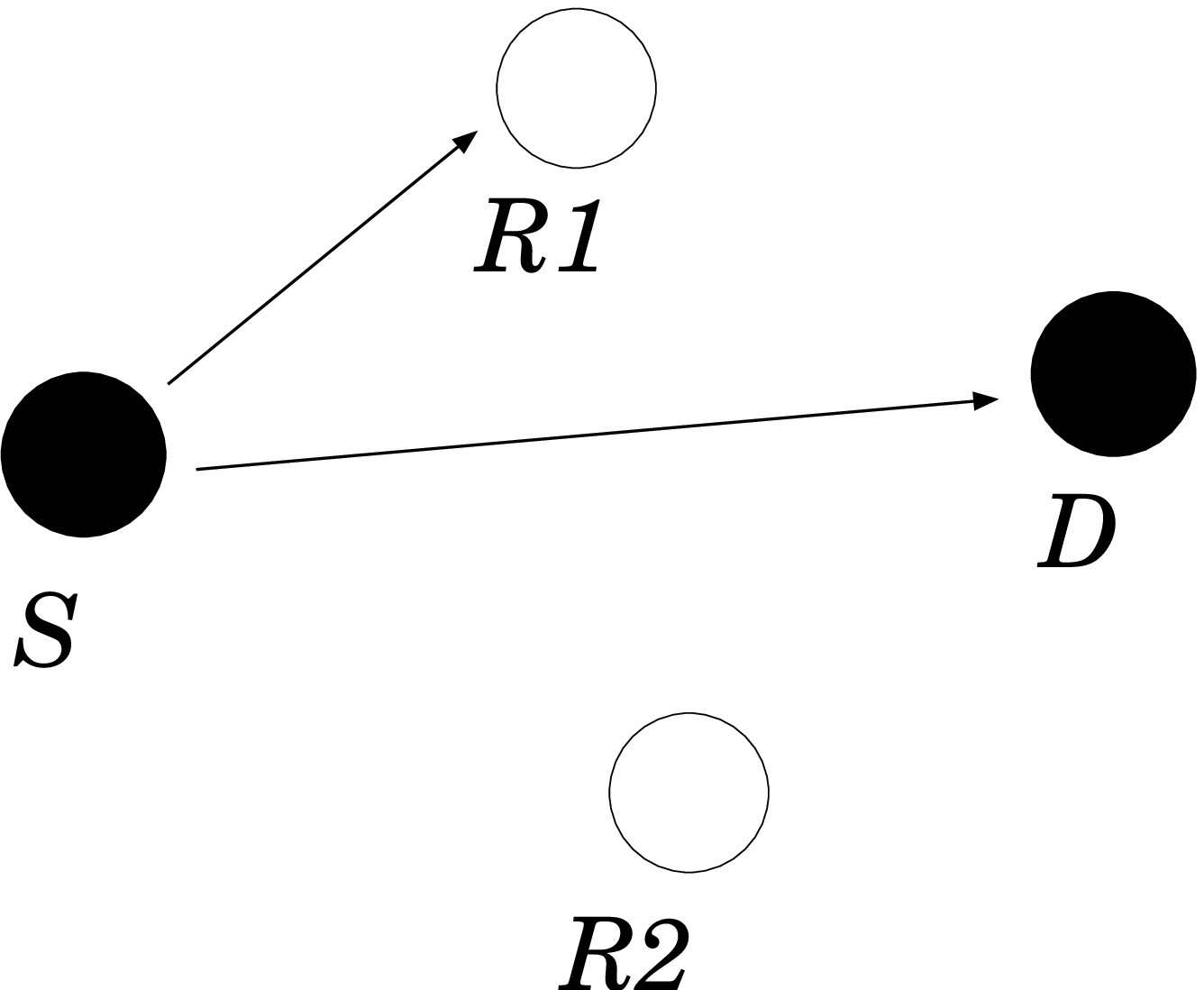}
\label{slot1}} \hfil \subfigure[Time slot
2.]{\includegraphics[width=2.5in]{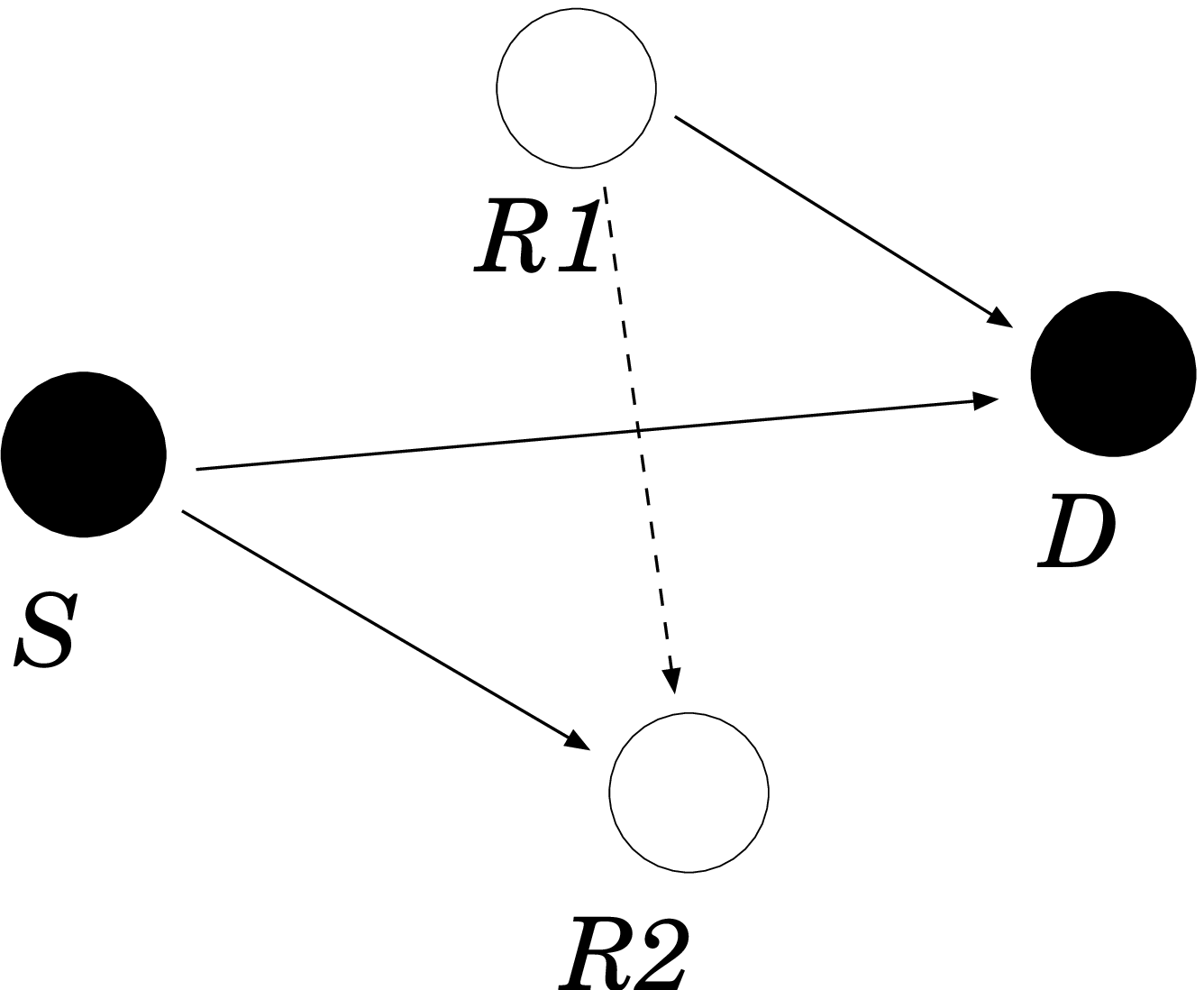}
\label{slot2}} \hfil \subfigure[Time slot
3.]{\includegraphics[width=2.5in]{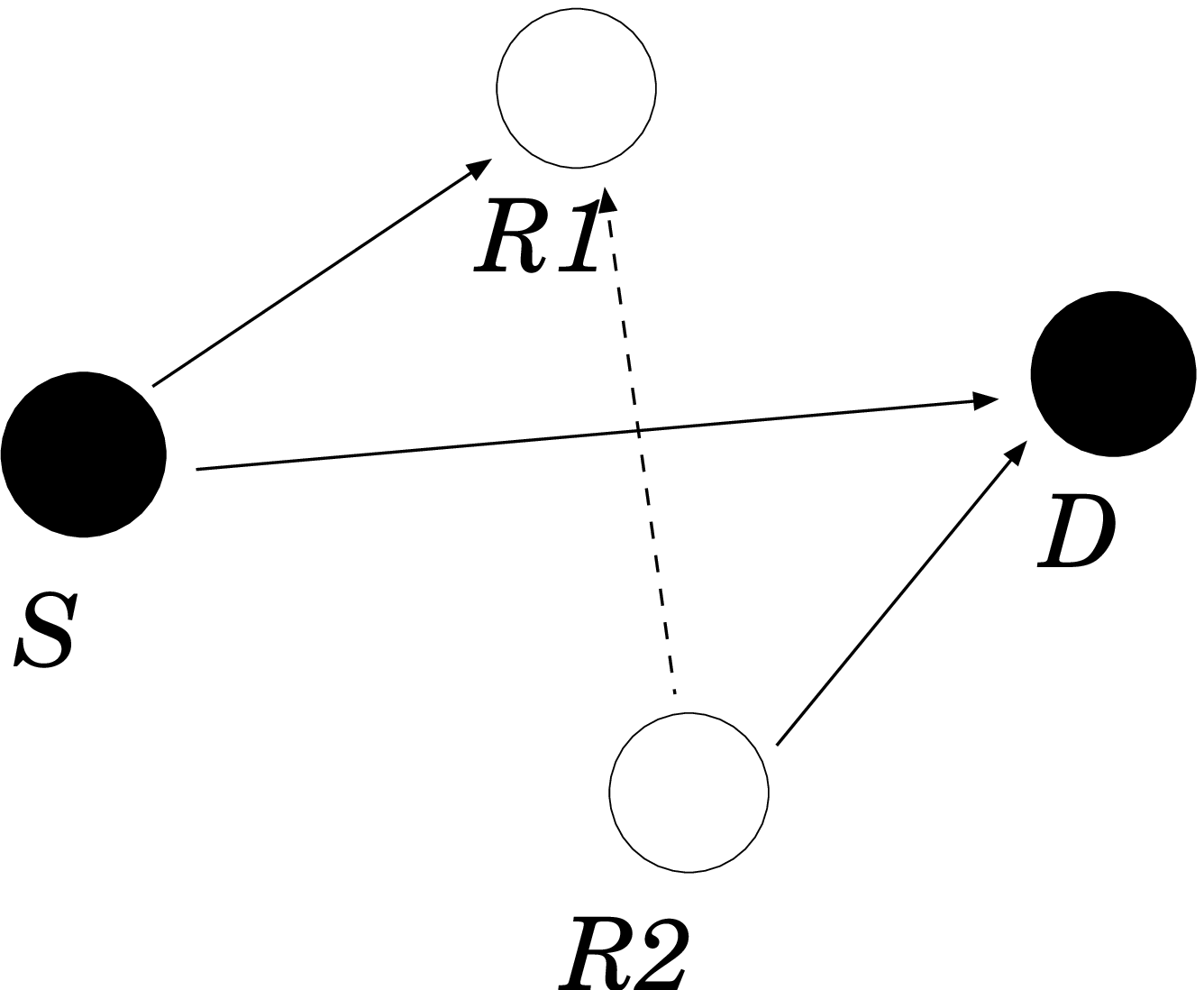}
\label{slot3}} \caption{Transmission schedule for the proposed
protocol.} \label{slot}
\end{figure}

\begin{table}[p!]
  \centering
\begin{tabular}{|c|c|c|}
  \hline
  Schemes/Maximum Gain & Multiplexing & Diversity \\
  \hline
  Direct transmission & 1 & 1 \\
  \hline
  Classic I & 1/3 & 3 \\
  \hline
  Classic II & 1/2 & 3 \\
  \hline
  Proposed scheme & $L/(L+1)$ & 2\\
  \hline
\end{tabular}
\caption{Comparison of the different transmission schemes for the
two relay case} \label{table2}
\end{table}

\begin{figure}[p!]
\centering
\includegraphics[width=3.5in]{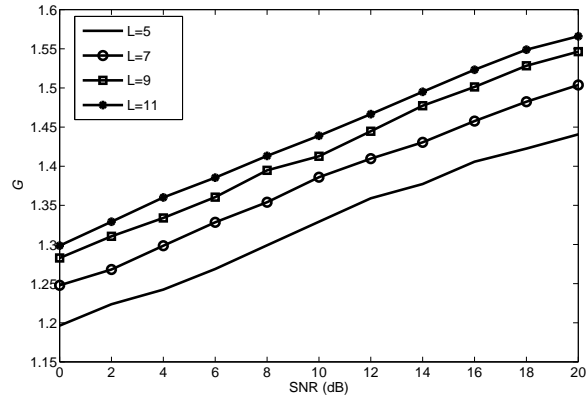}
\caption{Capacity gain of the proposed protocol over classic
protocol II.} \label{succgain}
\end{figure}

\begin{figure}[p!]
\subfigure[Case I]{\includegraphics[width=1.5in]{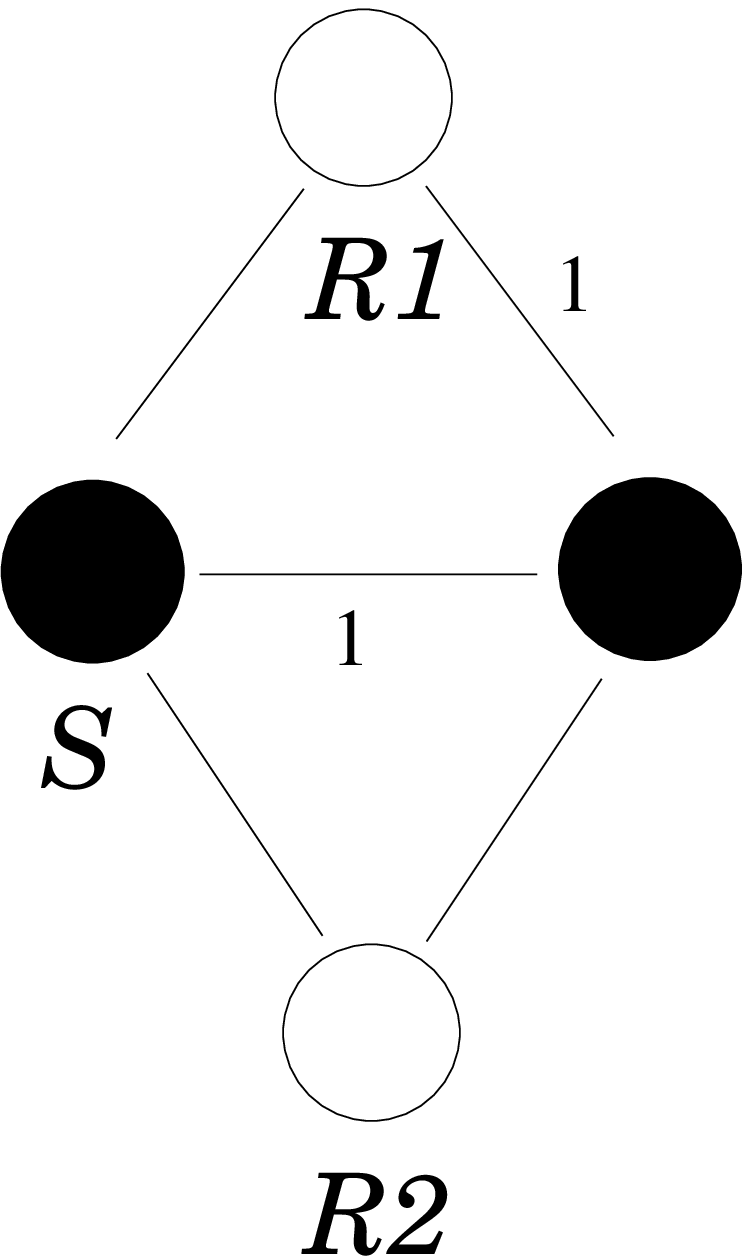}
\label{case1}} \hfil \subfigure[Case
II]{\includegraphics[width=1.8in]{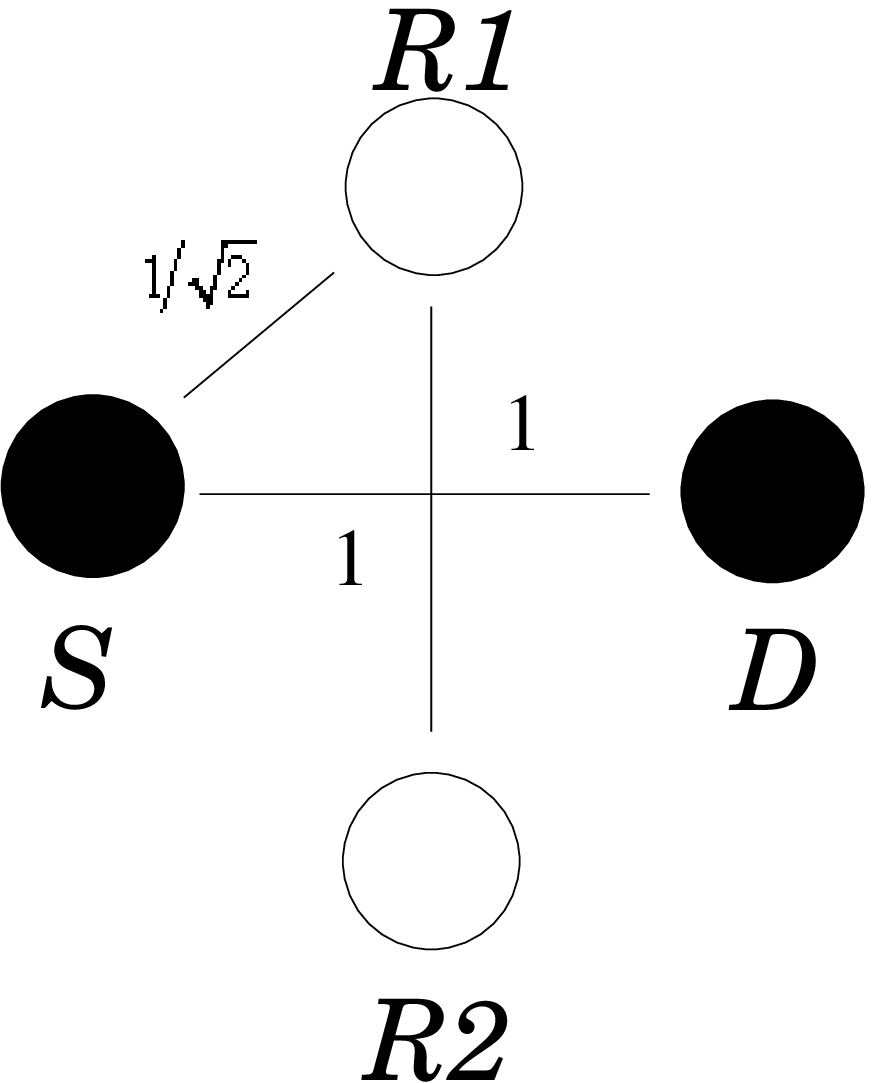}
\label{case2}} \hfil \subfigure[Case
III]{\includegraphics[width=2.5in]{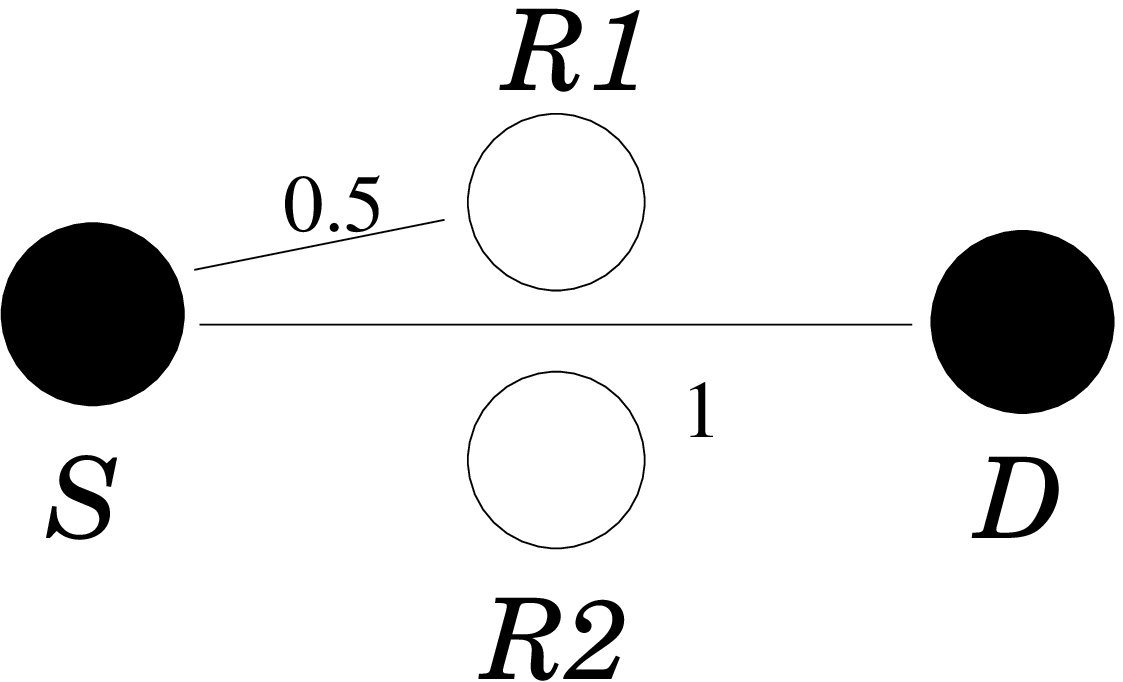}
\label{case3}} \caption{Network models for different geometries.}
\label{model}
\end{figure}

\begin{figure}[p!]
\subfigure[Case I]{\includegraphics[width=3.5in]{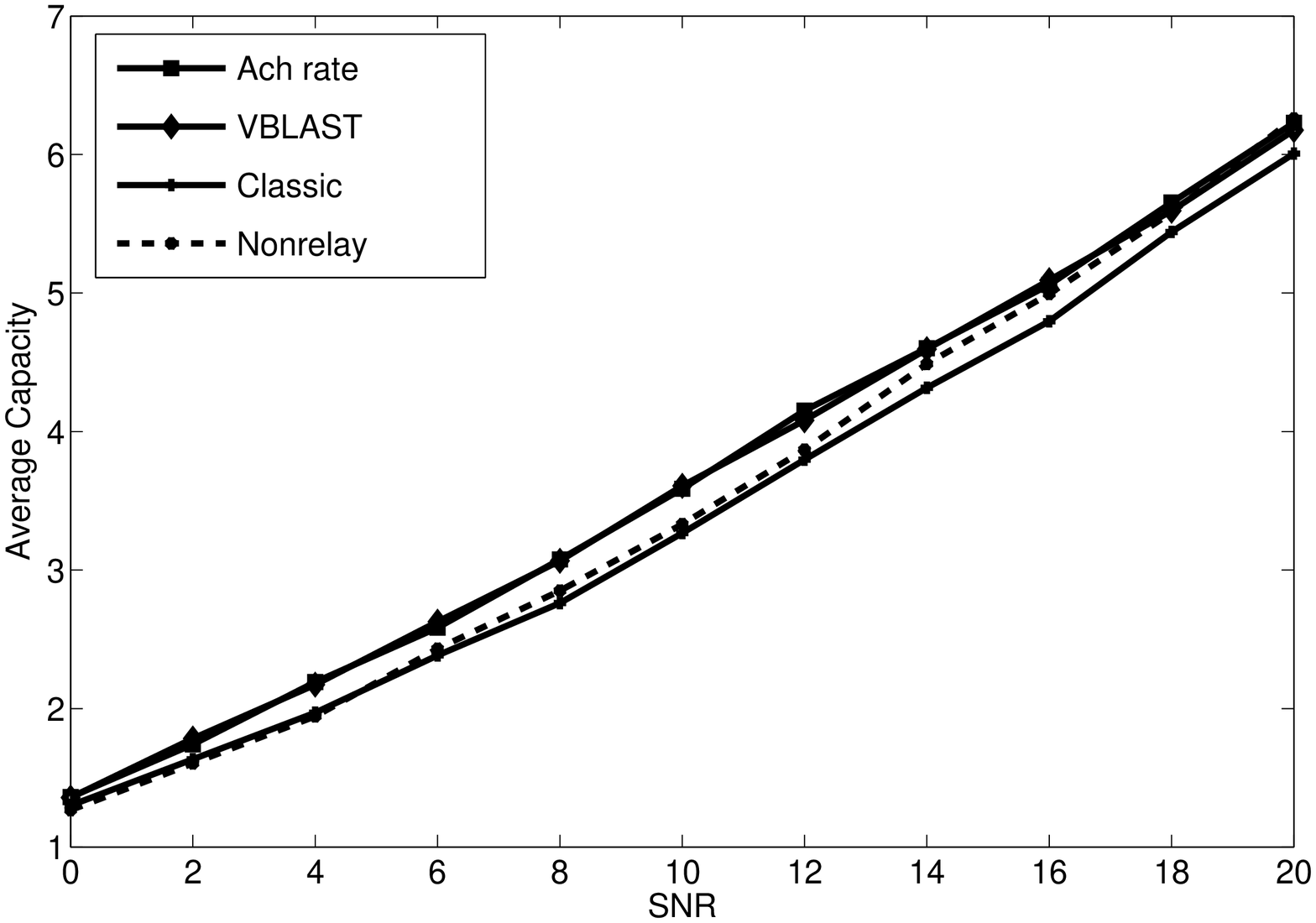}
\label{case1}} \hfil \subfigure[Case
II]{\includegraphics[width=3.5in]{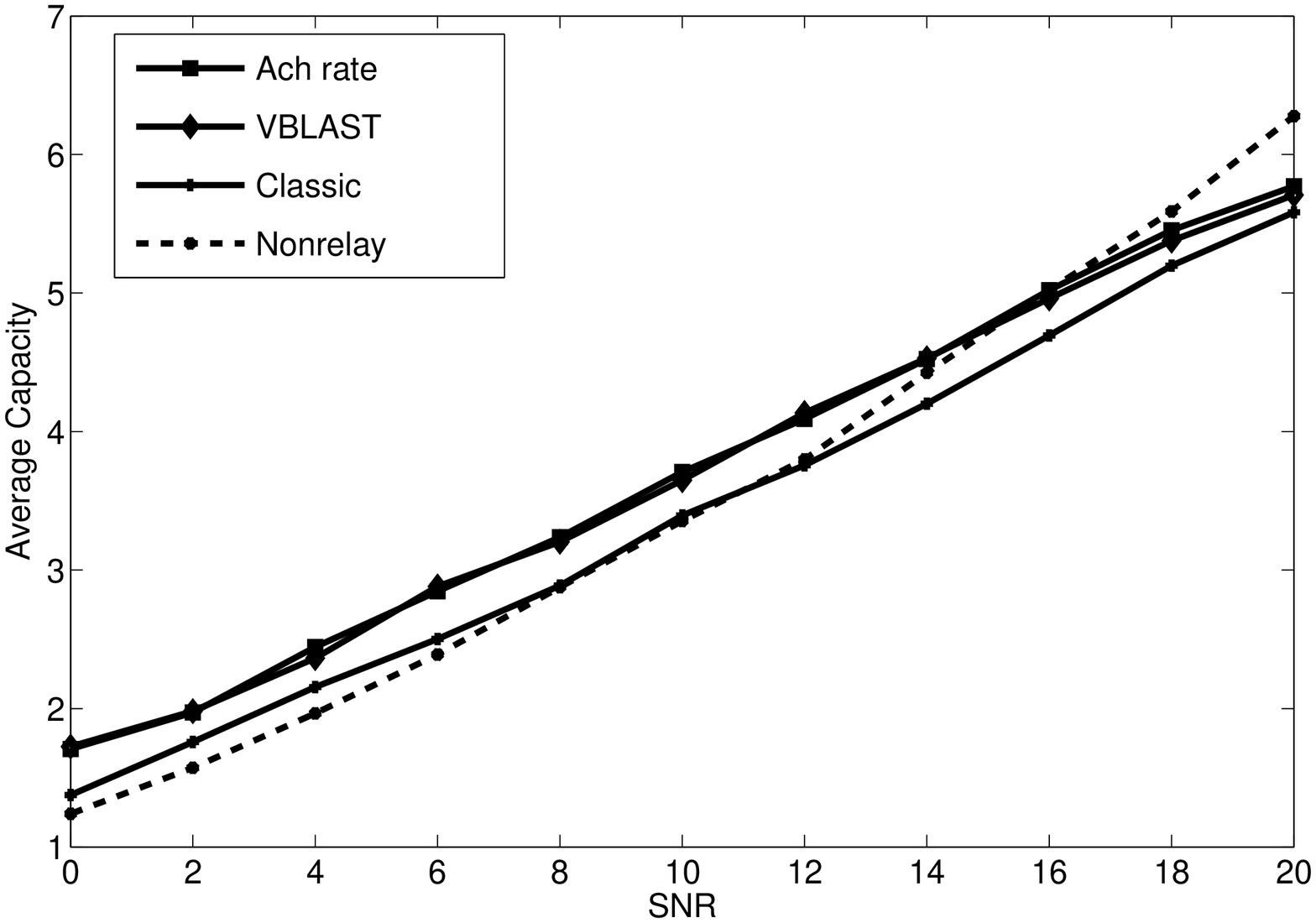}
\label{case2}} \hfil \subfigure[Case
III]{\includegraphics[width=3.5in]{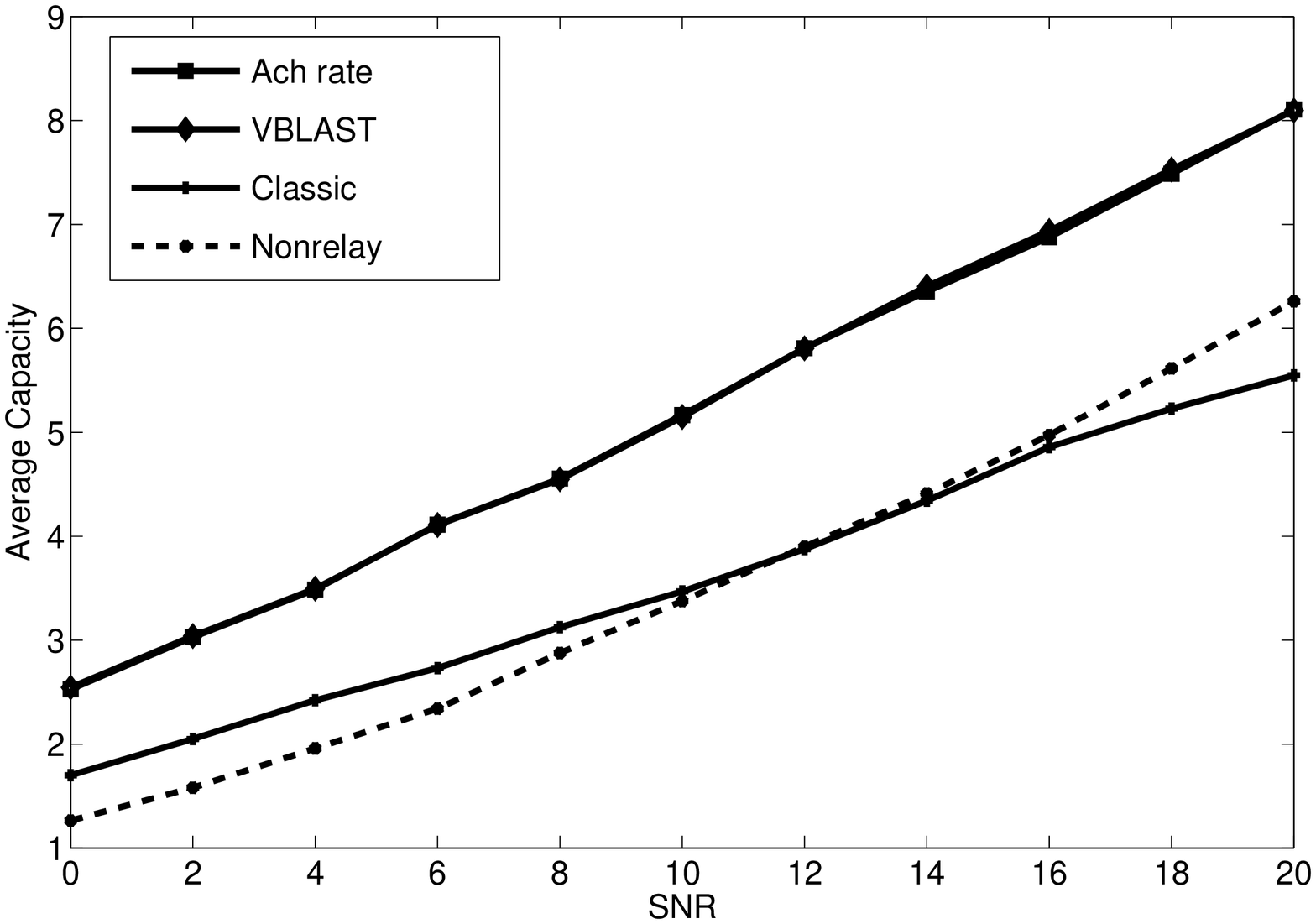}
\label{case3}} \caption{Average capacity of the network for
different network geometries in bits per transmission time slot.}
\label{sim}
\end{figure}






\end{document}